\documentclass[12pt]{article}

\usepackage{amssymb,amsmath,amsfonts,amssymb}
\usepackage{graphics,graphicx,color}
\usepackage{empheq}

\def\@abssec#1{\vspace{.05in}\footnotesize \parindent .2in
{\bf #1. }\ignorespaces}

\graphicspath{{/EPSF/}{}}
\DeclareGraphicsExtensions{.eps}

\newtheorem{theorem}{Theorem}[section]

\newtheorem{proposition}[theorem]{Proposition}

\def \Rm {\mathbb R}

\def \Cm {\mathbb C}

\def \Zm {\mathbb Z}
\def \Sm {\mathbb S}

\newcommand{\dsum}{\displaystyle\sum}
\newcommand{\dint}{\displaystyle\int}

\newcommand{\pdr}[2]{\dfrac{\partial{#1}}{\partial{#2}}}

\newcommand{\kp}{{\mathbf{k}}} 

\newcommand{\bk}{\mathbf k}

\newcommand{\bx}{\mathbf x} \newcommand{\by}{\mathbf y}

\newcommand{\mC}{\mathcal C}
\newcommand{\mF}{\mathcal F}

\newcommand{\mW}{\mathcal W}

\newcommand{\fC}{{\mathfrak C}}

\newcommand{\cout}[1]{}

\newcommand{\sgn}[1]{\,{\rm sign}(#1)}

\newcommand{\ow}{{\rm Op}^w}

\title{Quantum anomalous Hall phases in gated rhombohedral graphene}

\author{Matthew Frazier \thanks{Committee on Computational and Applied Mathematics, University of Chicago, Chicago, IL 60637; {\tt mjfrazier@uchicago.edu}} \and Guillaume Bal \thanks{CCAM, Departments of Mathematics and Statistics, University of Chicago, Chicago, IL 60637; {\tt guillaumebal@uchicago.edu}}  }

\begin{document}
 
\maketitle

\begin{abstract}
 We consider a coupled system of Dirac operators that models spin- and valley-polarized gated rhombohedral graphene (RHG) with an arbitrary number of layers. We classify all quantum anomalous Hall phases that are compatible with the model and show that a bulk-edge correspondence exists between bulk phases and chiral edge states carrying a quantized anomalous Hall current.  When the displacement field is sufficiently small compared to the interlayer coupling in the RHG application, we retrieve the known phases where the charge is given by the number of graphene layers. When the displacement field increases, we identify all possible topological phase transitions and corresponding quantized chiral edge charges. Numerical simulations confirm the theoretical findings.
\end{abstract}


\maketitle 


%
\section{Introduction}
\label{sec:intro}
Since the discovery of the quantum Hall effect \cite{klitzing1980new,PhysRevLett.49.405,avron1994,bernevig2013topological,prodan2016bulk}, it was recognized that many phases of matter in different fields of physics were topological in origin \cite{delplace2017topological,fu2022dispersion,PhysRevLett.100.013904,silveirinha2016bulk,lu2014topological}.  In topological insulators, a distinct manifestation of the topological nature of the phases is a transport asymmetry at interfaces separating insulators in different topological phases. This asymmetry is typically related to the bulk topological phases by a bulk edge correspondence (BEC) \cite{bernevig2013topological,prodan2016bulk}. We focus here on effective, macroscopic models of two dimensional systems. In that context, the BEC was demonstrated to hold for a class of effective elliptic problems, including those considered in this paper, in Refs \cite{bal2022topological,quinn2024approximations,bal2023topological,bal2024continuous}; see Refs \cite{drouot2021microlocal,silveirinha2016bulk,prodan2016bulk} for derivations of bulk-edge correspondences in different settings. A salient feature of the work in \cite{bal2022topological,quinn2024approximations} is the observation that bulk phase differences are more generally defined than differences of (possibly ill-defined) bulk phases; see also \cite{bal2019continuous,silveirinha2015chern,silveirinha2016bulk} for related works where bulk phases may be defined. The physical transport asymmetry is then related to a topological bulk-difference invariant (BDI), which remains to be computed for problems of interest. 

One such class of topological phases produce quantized electron transport referred to as the quantum anomalous Hall effect (QAHE). This theory, first proposed in a prototype model by Haldane\cite{PhysRevLett.61.2015}, describes quantized electron transport in 2d materials which, unlike in the quantum Hall effect, does not require an external magnetic field to be applied. A number of mechanisms have been used to realize the QAHE \cite{chang2023colloquium}, and one recent realization is in layered graphene systems with ABC-type, or rhombohedral, stacking. The band structure for effective models of multi-layered graphene systems was analyzed in a series of works \cite{min2008electronic,zhang2010band,zhang2011spontaneous, zhang2010spontaneous} and a first experimental result observing quantized Hall conductance in multi-layered graphene at zero magnetic field was obtained in \cite{han2024large} for a five layer rhombohedral graphene (RHG) system and \cite{choi2025superconductivity} for a 4-layer system. This system, assumed to be spin-polarized and valley-polarized, forms the main application for the effective model described in section \ref{sec:model}.

A second application may be found in the field of Floquet topological insulators (FTI). When a sheet of graphene is irradiated by a time-harmonic laser field, effective replica models may be obtained asymptotically in inverse powers of the laser frequency. These models have the same mathematical form as those in RHG applications. See \cite{Perez_2015,mciver2020light} for details on this application and \cite{bal2022multiscale} for models of the form \eqref{eq:H} considered below. For concreteness, we focus on the RHG application in this paper, though pertinent results from FTI's are cited where they coincide with the RHG model.

The effective Hamiltonian we consider in this paper models an arbitrary number of layers in gated RHG. It includes two main parameters, the interlayer coupling coefficient, and the difference of potential across the layers generated by a constant displacement field. The objective of this paper is a complete classification of the QAH phases as these parameters vary. For small values of the displacement field compared to the coupling constant, we retrieve a transport asymmetry given by the number of layers in the system. For large values of the displacement field, we retrieve results obtained in \cite{bal2022multiscale} in the context of FTI. 

As mentioned above, we use the notion of BDI's to define Chern numbers which topologically characterize the shared band gap of two bulk QAH phases, defined by different gate potentials or inter-layer couplings, which meet along an interface. We prefer this approach to the prevailing approach of defining band-by-band bulk invariants as integrals of Berry curvature for two important reasons. First, even as Haldane originally noted\cite{PhysRevLett.61.2015}, these bulk invariants are not gauge-invariant and thus predictions of absolute values of Hall conductance cannot only rely on the assignment of bulk invariants but must rely on other model considerations to fix a gauge. Second, integrals of Berry curvature are not guaranteed to be \textit{bona fide} Chern numbers unless the underlying connection is continuous (modulo a smooth gauge transformation) at every point on a closed manifold. For periodic models with a compact Brillouin zone these criteria are usually self-evident, however for continuum models we must consider a non-trivial compactification of the entire plane $\Rm^2\;$\cite{silveirinha2015chern,bal2023topological, bal2022topological}. Examples from equatorial waves \cite{tauber2019bulk}, continuum photonics\cite{silveirinha2015chern}, and the 2d Dirac model itself \cite{bal2019continuous} show that bulk invariants for continuous models often do not define \textit{bona fide} Chern numbers without regularizing terms, and indeed this is also the case for the RHG model. By considering BDI's instead of bulk invariants we eliminate the need to fix the correct gauge by considering differences in Berry curvatures and considerably simplify the construction of well-defined Chern numbers by providing much more flexible criterion for their definition (see Appendix \ref{sec:app} and \cite{bal2022topological}).

An outline for the rest of the paper is as follows. Section \ref{sec:model} presents the model Hamiltonians and the main results on their topological classification. We highlight the symmetries of the system, which greatly simplify the computation of Chern numbers \cite{silveirinha2016bulk,frazier2025topological}, while noting where simplifications in our model may provide obstructions to the observation of QAH phases in real systems. The derivation of these results is postponed to section \ref{sec:proof}, where we present a proof of our main theorem defining a BEC for the system. In section \ref{sec:num} we detail a numerical scheme used to approximate the spectrum of continuous Hamiltonians and present numerical simulations of the RHG interface Hamiltonians to provide a quantitative description of the edge modes for different values of the number of layers and the displacement field. Concluding remarks are given in section \ref{sec:conclu} while relevant results on the classification of elliptic Hamiltonians are recalled in Appendix \ref{sec:app}.

\section{Effective model and classification}
\label{sec:model}
\paragraph{Effective RHG model.}
We consider for an effective model of rhombohedral graphene with $m\geq2$ layers the following $2m\times 2m$ system of Dirac equations
\begin{equation}\label{eq:H}
 H = \begin{pmatrix} u_1+ v_0 D\cdot\sigma & B & 0 & \ldots & 0 \\
   B^* & u_2+ v_0D\cdot\sigma & B & \ddots &  \vdots \\
   0 & B^* & \ddots  & \ddots & 0 \\ 
   \vdots  & \ddots & \ddots & \ddots & B\\
   0 & \ldots & 0 & B^* & u_m+ v_0 D\cdot\sigma
   \end{pmatrix}
\end{equation}
which acts on state vectors $\psi\in H^1(\Rm^2;\Cm^{2m})$ for which the coordinates $(\psi_{2j-1}, \psi_{2j}), j \in \{1, ..., m\}$ describe a pseudo-spin coordinate for, respectively, the $A$ and $B$ sub-lattices of the $j$-th layer\cite{zhang2010band, zhang2011spontaneous}. Here $D\cdot\sigma=-i\partial_x \sigma_1 -i \tau \partial_y \sigma_2$ for $(x,y)\in\Rm^2$ parametrizing the two-dimensional material, $\sigma_{1,2}$ are standard Pauli matrices, and $\tau=\pm1$ is a valley index. For the rest of the paper, we assume the Fermi velocity $v_0$ constant and normalized to $1$ by an appropriate rescaling of the spatial variables. We also assume a valley index $\tau=1$ unless otherwise mentioned knowing that all invariants presented below should be multiplied by $\sgn{\tau}$. The real numbers $u_j$ for $1\leq j\leq m$ describe the electric potential (bias) at layer $j$.

The above effective, macroscopic model of RHG, is taken from \cite{min2008electronic} assuming only nearest-layer inter-layer coupling. The same system of equations is also used as the $N-$replica model of Floquet topological insulators considered in \cite{bal2022multiscale}. 

The matrix $B$ encodes inter-layer coupling and will be taken of the form
\begin{equation} \label{eq:B}
 B = \gamma A \quad \mbox{ or } \quad B=\gamma A^* \qquad \mbox{ for } \quad  A=\begin{pmatrix} 0 & 0 \\ 1 & 0\end{pmatrix}
\end{equation}
with $\gamma$ a parameter of the model, which in \cite{han2024large}(SM) has approximate value $\gamma=0.435\,$eV. Here, ${}^*$ denotes Hermitian transpose. The above choices of coupling correspond to ABC and CBA stacking, respectively- both rhombohedral stacking but with oppostite chirality. See \cite{min2008electronic} for a detailed discussion of graphene stacking orders.

Assuming a constant displacement field ${\rm D}=-\epsilon_0\nabla V$ generating a potential difference between layers $1$ and $m$ equal to $u$, the induced potential at layer $1\leq j \leq m$ is given by
\begin{equation}\label{eq:uj}
    u_j = \frac{u}{m-1}\big(j-\frac{m+1}2\big).
\end{equation}
We note that $u_{m+1-j}=-u_j$ for $1\leq j\leq m$, which will allow us to show that the model satisfies particle-hole symmetry. Note also that the Hamiltonian $H$ is tridiagonal when $B=\gamma A$, which we assume from now on unless mentioned otherwise. 
\medskip
\paragraph{Bulk and Interface Hamiltonians.} The bulk Hamiltonians $H_B$ are defined as the model $H$ of \eqref{eq:H} with all coefficients constant. We denote the two bulk operators $H_B^{N/S}$ corresponding to a gate potential $u = u^N$ for $H_B^N$ and $u =u^S$ for $H_B^S$ while all coefficients $u_j$ are given by \eqref{eq:uj} ($N/S$ standing for North/South). We consider $\gamma$ to be a given constant throughout the rest of the paper.

The interface Hamiltonian $H_I$ is defined as $H$ in \eqref{eq:H} with now $u_j=u_j(y)$. We model a transition from one given displacement field generating the potential $u(y)=u^N$ for $y\geq R>0$ to another displacement field generating $u(y)=u^S$ for $y\leq -R$, assuming that $u(y)$ is smooth in the interval $y\in [-R, R]$. 

Following \eqref{eq:weylquantization} in the appendix, the Weyl symbol of $H_I$ is given by
\[ a(x,y,k_x,k_y) = \begin{pmatrix} u_1(y)+ v_0 \kp\cdot\sigma & B & 0 & \ldots & 0 \\
   B^* & u_2(y)+ v_0\kp\cdot\sigma & B & \ddots &  \vdots \\
   0 & B^* & \ddots  & \ddots & 0 \\ 
   \vdots  & \ddots & \ddots & \ddots & B\\
   0 & \ldots & 0 & B^* & u_m(y)+ v_0 \kp\cdot\sigma
   \end{pmatrix}\]
with $\kp\cdot\sigma=k_x\sigma_1+k_y\sigma_2$. Our results in Theorem \ref{thm:gap} will show that the symbol satisfies hypothesis [H1] in \cite{quinn2024approximations,bal2022topological} as an elliptic operator of order $m=1$ that is gapped for $|y|$ sufficiently large; see the appendix for more details. The Hamiltonian $H_I$ therefore models a transition about $y\approx0$ between two (topological) insulators for $|y|\geq R$.

Our objective is to propose a classification of $H_I$ based on the transport asymmetry observed along the interface $y\approx0$. This quantized asymmetry is independent of the profile $u(y)$for any given values of $u^N, u^S$. 

\medskip
\paragraph{Symmetries of bulk Hamiltonians.} The constant-coefficient bulk Hamiltonians admit the spectral representation in the Fourier variables
\begin{equation}\label{eq:HB}
  H_B = \mF^{-1} \hat H_B(\kp) \mF
\end{equation}
where $\mF$ is Fourier transform and $\hat H_B(\kp)$ is a $2m\times 2m$ Hermitian matrix where $D_x$ and $D_y$ are replaced by $k_x$ and $k_y$ with $\kp=k_x+i k_y\in\Cm$ (identified with $(k_x,k_y)\in\Rm^2$ whenever necessary). The spectrum of $H_B$ is then given by $2m$ branches of (absolutely continuous) spectrum $\bk\mapsto E_j(\bk)$ parametrized by $\kp\in\Cm$. These branches are in fact analytic as we will show that they are simple (see  \cite[Theorem VII.1.7 and Section VII.3.1]{kato2013perturbation}). 

The bulk Hamiltonians satisfy a number of important symmetries. We recall that $H_B$ is defined in \eqref{eq:H} with $u_j$ given in \eqref{eq:uj} for $u\not=0$ and $B=\gamma A$ for $\gamma\not=0$. We assume a valley index $\tau=1$.
Let $\kp=ke^{i\theta}$ for $k\geq0$ and $\theta\in[0,2\pi)$ and observe that
\[  \mF D\cdot \sigma \mF^{-1} = k_x\sigma_1 + k_y \sigma_2 = \begin{pmatrix}  0 & \kp^* \\ \kp & 0\end{pmatrix} = 
k \begin{pmatrix}  0 & e^{-i\theta} \\ e^{i\theta} & 0\end{pmatrix}.\]
Define the unitary (diagonal) matrix:
\begin{equation}\label{eq:U}
  U(\theta) = {\rm Diag} (1, e^{i\theta},e^{i\theta},e^{2i\theta},\ldots,e^{i(m-1)\theta},e^{im\theta}).
\end{equation}
Define the $2m\times 2m$ matrices $\Gamma_j={\rm Diag}(\sigma_j)$ for $j=1,2,3$ and $G_1$ the matrix with $\sigma_1$ on the (block) antidiagonal. We denote by $\hat H_B(\kp,\gamma,u)$ the operator $\hat H_B(\bk)$ constructed with parameters $(\gamma,u)$ and $\hat H_B(k, \gamma, u)$ to be $\hat H_B(\kp,\gamma,u)$ with $\theta = 0$ so that $\bk = |\bk| = k \in \Rm$. 

\begin{proposition} \label{prop:symmetries}
 Let $\hat H(\kp)=\hat H_B(\kp)$ be defined as above with $\gamma\not=0$ and $u\not=0$. Then 
 \begin{eqnarray}
   U^*(\theta) \hat H(\kp, \gamma, u) U(\theta) &=& \hat H(k, \gamma, u),\quad
  G_1 \hat H(k,\gamma, u) G_1 = \hat H(k,\gamma, -u), \\
  G_1  \Gamma_3 (-\hat H(k, \gamma, u)) \Gamma_3 G_1 &=& \hat H(k, \gamma, u),\quad \Gamma_3 \hat H(k,\gamma, u) \Gamma_3 = H(-k,-\gamma, u).
 \end{eqnarray}
\end{proposition}
The proof of this result is a simple verification using that $-\hat H_B(k,\gamma,u)=\hat H_B(-k,-\gamma,-u)$. This shows that $\hat H_B(\bk, \gamma, u)$ and $-\hat H_B(\bk,\gamma, u)$ are unitarily equivalent implying the parity relation $E_{2m+1-j}(\kp)=-E_j(\kp)$. These relations also show that the branches of spectrum of $H_B$ are invariant by rotation.

Note that the invariance by rotation and the particle-hole invariance are properties of the simplified model \eqref{eq:H}. In the presence of more general coupling terms, for instance when $B$ is a full matrix, or when interlayer couplings involve more distant layers \cite{han2024large,zhang2010band,zhang2011spontaneous}, then these properties no longer always hold. 

Although we assume for most of our analysis that $\tau=1$ and $B=\gamma A$ we also note that
\begin{eqnarray}\label{eq:symtauA}
\hat H_B(\kp, \tau=-1) &=& \hat H_B(\kp^*, \tau=1) = U^*(\theta) \hat H_B(k, \tau=1) U(\theta)\\
\nonumber
\hat H_B(\kp, B =\gamma A^*) &=& \Gamma_1 \hat H_B(\kp^*,B = \gamma A)\Gamma_1 = \Gamma_1 U^*(\theta) \hat H_B(k, B =\gamma A) U(\theta)\Gamma_1.
\end{eqnarray}
The spectrum of $\hat H_B$ is therefore invariant with respect to the map $A\to A^*$ (corresponding to a CBA stacking rather than an ABC stacking). The unitary transformation from $\hat H_B(\kp)$ to $\hat H_B(k)$ encoded by $U(\theta)$ in \eqref{eq:U} is therefore replaced by $U^*(\theta)$ when $\tau=-1$ and $B=\gamma A$ and by $\Gamma_1 U^*(\theta)$ when $\tau=1$ and $B=\gamma A^*$. Conjugation by $\Gamma_1 U(\theta)$ similarly maps $\hat H(\kp, \tau=-1,\gamma A^*)$ to $\hat H(k, \tau=-1,\gamma A^*)$. These unitary equivalences will allow us to extend the BEC to both stacking orders and valley indices without additional analysis.
\medskip
\paragraph{Gapped phases.}
From the symmetries of the bulk Hamiltonians, we deduce the following property on the branches of spectrum of the bulk Hamiltonians.

\begin{theorem}[Gapped Hamiltonian]\label{thm:gap}
 Let $\hat H_B^{h}(\kp)$ be defined as above for $\gamma>0$ and $|u|>0$ with $h\in\{N,S\}$. Let $k=|\kp|>0$. Then there is $0<E_0=E_0(k, \gamma,u)$ such that $\hat H_B^{h}(\kp)$ is gapped in $(-E_0,E_0)$. Moreover, all eigenvalues of the operator $\hat H_B^h(\kp)$ are simple.
\end{theorem}
\begin{proof}
    By rotational invariance we may assume that $\kp=k$ for $k>0$. Then, $\hat H(k)$ is a symmetric (Jacobi) tridiagonal matrix with positive off-diagonal entries given by $k>0$ and $\gamma>0$. We know that their eigenvalues are simple (since the first component of any eigenvector uniquely determines the other components iteratively). Since $E_{2m+1-j}(\kp)=-E_j(\kp)$ and $|E_j(\kp)|\to\infty$ as $|\kp|\to\infty$ (by ellipticity of $H_B$), we deduce that $E_j(\kp)=0$ is not possible unless  $\kp=0$ and hence the existence of a gap $E_0(k, \gamma, u)>0$. 
\end{proof}
\medskip

This result shows that gap closing at $E=0$ can only occur when $\bk=0$. Such closings do occur and induce (bulk and edge) phase transitions. When no such gap closing occurs at $\kp=0$, then the spectral gap is global in the sense that there exists $0<E_0(\gamma,u)$ such that the spectrum $\hat H^h_B(\kp)$ does not intersect $(-E_0,E_0)$ for any $\kp$ (by a compactness argument since $|E_j(\kp)|\to\infty$ as $|\kp|\to\infty$ and $\kp\to E_j(\kp)$ is continuous; in fact real-analytic; (see \cite[Theorem VII.1.7 and Section VII.3.1]{kato2013perturbation})).

For such values of $(u^h,\gamma)$, we therefore unambiguously define the projectors
\begin{equation}\label{eq:proj}
    \Pi^h(\kp) = \chi(\hat H^h_B(\kp)<0)
\end{equation}
for $h\in\{N,S\}$ projecting onto the $m-$dimensional vector space spanned by the eigenvectors associated to the negative eigenvalues of $\hat H^h_B(\kp)$. Recall that we assume $\gamma$ does not vary, while $\hat H_B^N(\bk)$ is characterized by a gate potential $u^N$ and $\hat H_B^S(\bk)$ is characterized by a gate potential $u^S$. 

We define the bulk-difference invariant (BDI) \cite{bal2022topological} as follows:
\begin{equation}\label{eq:ChernBDI}
    {\rm BDI}(\gamma, u^S, u^N)=  \fC[\Pi^S,\Pi^N]   := \frac  i{2\pi} \int_{\Rm^2} {\rm tr} \Pi^S d\Pi^S\wedge d\Pi^S
 - \frac  i{2\pi} \int_{\Rm^2} {\rm tr} \Pi^N d\Pi^N \wedge d\Pi^N.
\end{equation}
Although there are no general guarantees that integrals of this form are well-defined, given appropriate assumptions on $\Pi^N, \Pi^S$ (e.g. \eqref{eq:gluing} in the following sections) the invariant may be interpreted as the Chern number of the family of projectors $\{\Pi^N,\Pi^S\}$ defined by stereographic projection of $\Rm^2$ onto the unit sphere and is thus guaranteed to take values in the integers by their interpretation as Atiyah-Singer indices\cite{bal2022topological, bal2023topological}. By contrast, for the operators considered in this paper the integrals of the Berry curvature $\Pi^h d\Pi^h\wedge d\Pi^h$ appropriately normalized are not guaranteed to be integer-valued \cite{bal2019continuous}, and even when they are, are not gauge-invariant. Therefore, whereas absolute phases for the N/S insulators may not be defined unambiguously, phase differences as in \eqref{eq:ChernBDI} are indeed well-defined; see Ref \cite{bal2022multiscale,bal2024continuous} and the appendix for more detail.

The Hamiltonian $H_I$ and corresponding bulk insulators $H^{N/S}_B$ satisfy the ellipticity conditions of \cite{bal2022topological,quinn2024approximations}. As a consequence, the bulk-edge correspondence \cite{bal2022topological,quinn2024approximations} applies and the asymmetry of electron transport governed by $H_I$ is then exactly given by the above BDI. We refer to the Appendix for a summary of results leading to the definition of \eqref{eq:ChernBDI}, of the edge invariant $\sigma_I[H_I]$ quantifying total electron transport along the interface $y \approx 0$, and the bulk-edge correspondence stating that $2\pi\sigma_I[H_I]={\rm BDI}(\gamma,u^S, u^N)$.
\medskip
\paragraph{Classification of QAH phases.} It remains to compute the BDI as a function of $(\gamma,u^S, u^N)$ provided that $E_0(\gamma,u)>0$ and to identify the values of $(\gamma,u)$ where phase transitions may occur. The number of possible such phases depends on the number of layers $m$. 

At $k=0$, the matrix $\hat H_B(\bk)$ splits into $1\times1$ or $2\times2$ blocks. The latter are of the form $(u_j,\gamma;\gamma,u_{j+1})$ for $1\leq j\leq m-1$. The eigenvalues of these blocks cross $E=0$ when $u_ju_{j+1}=\gamma^2$. Assuming a constant displacement field, this is the constraint 
\begin{equation}\label{eq:delta}
 \frac{ \gamma^2(m-1)^2 }{u^2} = (j-\frac m2)^2 - \frac14,\qquad \pm\big( j-\frac m2\big) = \delta(\gamma,u) := \sqrt{\frac{\gamma^2(m-1)^2 }{u^2}  + \frac 14}.
\end{equation}
When this constraint is satisfied each eigenvalue $E=0$ is degenerate of multiplicity $2$ since $u_{m+1-j} = -u_j$. Since $\gamma>0$ implies $\delta(\gamma, u)>\frac12$, we deduce that $\hat H_B(\kp)$ has a 2-fold degenerate eigenvalue $E = 0$ for critical points $u$ for which $\delta(\gamma, u)$ takes the values $j-\frac m2$ for $\lceil\frac m2+1\rceil\leq j\leq m-1$. This is $\frac12(m-3)$ values when $m$ is odd and $\frac m2-1$ values when $m$ is even, or $\left\lfloor \frac m2-1\right\rfloor$ in both cases. Increasing $u>0$ from a small value compared to $\gamma$, we thus observe a first transition at $\delta(\gamma, u)=1$ when $m=4$ and $\delta(\gamma, u)=\frac32$ when $m=5$. 

Next define 
\begin{equation}\label{deltaj_def}
\delta_k:=k + \lceil \frac m2\rceil -\frac m2 = 
\begin{cases}
    k & m \;\text{even}\\
    k + \frac 12 & m\;\text{odd}
\end{cases}
\end{equation}
for $0\leq k\leq  \lfloor \frac m2 \rfloor$. For cases we consider below where $k< 0$ we define $\delta_k = \delta_{-k}$. Setting $j = k + \lceil\frac m2\rceil$ and rearranging \eqref{eq:delta} in terms of $\delta_k$ we deduce that phase transitions happen at critical points of potential $u$, assuming $\gamma$ remains constant:
\begin{equation}\label{eq:ucrit}
u_{\pm k} = \pm \gamma \frac{m-1}{\sqrt{\delta_{k}^2-\frac 14}}
\end{equation}
for $1 \le k \le \lfloor \frac m2 -1\rfloor$. Noting that $u = 0$ is also a critical value (where in fact all eigenvalues coincide at 0 when $\bk  =0$)\cite{min2008electronic}, there are therefore $2\lfloor \frac m2\rfloor$ distinct bulk topological phases which we label by $j \in \{\pm 1,\pm2, ...,\pm\lfloor\frac m2\rfloor\} $ defined by:
\begin{equation}\label{eq:phase_def}
    \hat H_B(\bk, \gamma, u) \in \text{sgn}(u)j \quad \text{iff}\quad \delta_{j-1} < \delta(\gamma, u) < \delta_j; \;\;1\le j\le \lfloor \frac m2\rfloor-1
\end{equation}
and $\hat H_B(\bk, \gamma, u) \in \text{sgn}(u) \lfloor\frac m2\rfloor$ iff $\delta(\gamma, u) > \delta_{\lfloor \frac m2\rfloor-1}$ (phase $\lfloor \frac m2\rfloor$ has no upper bound in $\delta(\gamma, u))$.

We thus observe the existence of only two phases when $m=2$ or $m=3$. Assuming $m\geq4$, the smallest value of $u$ at which we may observe a transition for a fixed $\gamma$ is when $k=\lfloor\frac m2 -1\rfloor$ leading to first critical value for the displacement field of
\[
u_{\lfloor\frac m2 -1\rfloor} = \gamma \frac{2(m-1)}{\sqrt{(m-3)(m-1)}} . 
\]
For $m$ large, this is asymptotically $u\approx 2\gamma$, which is significantly larger than the values displayed in current experiments \cite{han2024large, choi2025superconductivity}, which are constrained to around $u\approx 0.221$ eV and $u\approx 0.209$ eV respectively by a sufficient spin-orbit coupling (SOC) ensuring an appropriate spin polarization (recall that $\gamma$ is estimated to be 0.435 eV \cite{han2024large}). The largest critical value of $u$ is given by \eqref{eq:ucrit} with $\delta_k=1$ when $m\geq4$ is even and $\delta_k=\frac32$ when $m\geq5$ is odd so that the largest transition value of $u$ is proportional to $m\gamma $.

We are now ready to state the main result of the paper, the proof of which is delayed to the next section:
\begin{theorem}[BDI for QAH phases] \label{thm:class}
 Suppose that $\hat H^N_B$ is in phase $j$ and $\hat H_B^S$ is in phase $k$ according to \eqref{eq:phase_def} and let BDI$\,(\gamma,u^S, u^N)$ be defined as in \eqref{eq:ChernBDI} with $m\geq2$. Then the BDI$(\gamma, u^S, u^N)$ for a transition between the two bulk systems $\hat H_B^S$ and $\hat H_B^N$ is given by:
 \begin{equation}\label{eq:BDIj}
{\rm BDI}_{kj} = \begin{cases}
    \text{sgn}(j)\left(\delta_{k}(\delta_{k}-1)-\delta_{j}(\delta_{j}-1)\right) & \text{sgn}(j) = \text{sgn}(k)\\
    \text{sgn}(j)\left(\frac {m^2}2-(\delta_{j}(\delta_{j}-1)+\delta_{k}(\delta_{k}-1))\right) & \text{sgn}(j) = -\text{sgn}(k).
\end{cases}
\end{equation}
 for $j, k\in \{\pm1, \pm2, ...,\pm \lfloor \frac m2\rfloor\}$ and $\delta_j$ given by \eqref{deltaj_def}.
\end{theorem}

Thus, when $j=-k =\lfloor \frac m2\rfloor$ and $\delta_{\lfloor \frac m2\rfloor}=\frac{m}2$ for $m\geq3$, we find a minimal value for the BDI, BDI$_{(-\lfloor \frac m2\rfloor,\lfloor \frac m2\rfloor)}=m$ for $\frac{m}2-1<\delta(\gamma, u^{N/S})$, i.e., for $u^N$ and $u^S$ sufficiently small. This is the setting considered in \cite{han2024large, choi2025superconductivity, zhang2010band, zhang2011spontaneous} with $m=5$ and $m = 4$ respectively for the experimental results. On the other hand, when $u^N$ and $u^S$ are so large that $\delta_0 \le \frac12<\delta(\gamma, u^{N/S}) <\delta_1$, then the QAHE is maximized for $j=-k =1$ by $\frac12 (m^2-3)$ when $m$ is odd and by $\frac12 m^2$ when $m$ is even. The result for $m=2n+1$ corresponds to the $n$-replica model analyzed in \cite{bal2022multiscale}.

Keeping in mind that bulk phases with gate potential $-u$ and $u$ are energetically equivalent we consider the transition between phases $-k$ and $j$ to be equivalent to the transition between $-j$ and $k$ (note the equivalence of \eqref{eq:BDIj} for these two these transitions). Also note that ${\rm BDI}_{kj} = -{\rm BDI}_{jk}$. Disregarding these degeneracies we find $\lfloor\frac m2\rfloor^2$ distinct transitions between QAH phases. For concreteness, Figure \ref{fig:predict} enumerates all possible absolute values of BDI's according to \eqref{eq:BDIj} for a range of values $2\le m \le 9$.
\begin{figure}[h!]\small
\begin{tabular}{ccccccccccccccccc}
 \textbf{m =}&\textbf{2$^\dag$}\\
 &\textbf{3$^\dag$}\\
 &2 & \textbf{4$^\dag$} & 6 & 8$^\dag$\\
 &3 & \textbf{5$^\dag$} & 8 & 11$^\dag$\\
 &2 & 4 & \textbf{6$^{*\dag}$} & 10 & 12 & 14$^\dag$ &  16 & 18$^\dag$\\
 &3 & 5 & \textbf{7$^\dag$} & 8 & 12 & 15 & 17$^\dag$ & 20 & 23$^\dag$\\
 &2 & 4 & 6* & \textbf{8$^\dag$} & 10 & 12 & 14 & 18 & 20$^{*\dag}$ & 24 & 26 & 28$^\dag$ & 30 & 32$^\dag$\\
 &3 & 5 & 7& 8 & \textbf{9$^\dag$} & 12 & 15 & 16 & 21 & 23$^\dag$ & 24 & 28 & 31 & 33$^\dag$ & 36 & 39$^\dag$
\end{tabular}
\caption{All possible values of the BDI for the QAHE with $2\le m\le 9$. The symmetric case in which $u^N = -u^S > 0$ is denoted by $^\dag$; see Figure \ref{fig:5-layer_phase}. Bold numbers in each row reflect the number of layers \textbf{m} in each row (which is always a BDI). An * reflects that there is a degeneracy of two topological transitions with the same BDI which are \textit{not} energetically equivalent.}
\label{fig:predict}
\end{figure}

The values of the invariants in \eqref{eq:BDIj} were computed for a valley index $\tau=+1$ and ABC stacking $B=\gamma A$.  Using \eqref{eq:symtauA} and the fact that $U'(0)$ in \eqref{eq:R} below is replaced by $-U(0)'$ when $U(\theta)$ is replaced by $U^*(\theta)$, we find that the above invariants satisfy
\[
{\rm BDI}_{kj}[1,\gamma A] = - {\rm BDI}_{kj}[-1,\gamma A] = - {\rm BDI}_{kj}[1,\gamma A^*]= {\rm BDI}_{kj}[-1,\gamma A^*],
\]
where we use the notation ${\rm BDI}_{kj}={\rm BDI}_{kj}[\tau,B]$. For a given expression for $B$, we thus obtain that ${\rm BDI}_j[\tau=1]+{\rm BDI}_j[\tau=-1]=0$, implying that the topological invariant accounting for both valleys $\tau=\pm1$ is always trivial in the model we consider here.

\section{Proof of Theorem \ref{thm:class}.}
\label{sec:proof}

In this section we prove Theorem \ref{thm:class}. Following \eqref{eq:ChernBDI}, the derivation is based on computing the integral of Berry curvature associated to $\hat H^h(\bk)$ for $h=N/S$. Let $\hat H(\kp)=\sum_j E_j(k) \Pi_j(\kp)$ denote one of them. Following Theorem \ref{thm:gap}, the spectrum of this Hamiltonian is invariant by rotation with $E_j(k)$ simple eigenvalues, i.e., associated to $\Pi_j(\kp)$ a rank-one projector, and so that $E_{2m+1-j}(k) = -E_j(k)$ for $1\leq j\leq 2m$. In this section we take the view that $\bk \in \Rm^2$ to be consistent with \eqref{eq:ChernBDI} and the supporting theory, and identify $(k, \theta)$ as the polar coordinates of $\bk \in \Rm^2$.
\paragraph{Simplified formula for rotationally symmetric Hamiltonians.} Computing the integrals appearing in \eqref{eq:ChernBDI} analytically drastically simplifies in the presence of rotational symmetry. The Chern number of the projectors $\Pi^h(k, \theta)$ (recall \eqref{eq:proj}) is additive and hence may be written as the sum over branches $E_j(k)<0$
 of the rank-one projectors $\Pi_j(k, \theta)$ onto the eigenspace associated to $E_j(k)$\cite{bernevig2013topological}.

Assume $\Pi_j(k, \theta)=|\psi_j(k, \theta)\rangle\langle\psi_j(k, \theta)|$ is such a rank-one projector. We then verify that
\[
 {\rm tr} \Pi_j d\Pi_j \wedge d\Pi_j = d A_j,\qquad A_j(k, \theta) = (\psi_j(k, \theta),d\psi_j(k, \theta))
\]
where $A$ is a one-form-$i\Rm$-valued (Berry) connection assuming $d\psi_j(k, \theta)$ continuously defined, which can always be achieved since $\Rm^2$ is contractible. As an application of the Stokes theorem, we thus observe \cite{bernevig2013topological,silveirinha2015chern,hanson2016notes} that
\begin{equation}\label{eq:Stokes}
     \mC[\Pi_j] :=\frac i{2\pi}\int_{\Rm^2}dA_j(k, \theta) =  \frac{i}{2\pi}\Big(\oint_{k \rightarrow \infty} A_j(k, \theta)  - \oint_{k\rightarrow  0}  A_j(k, \theta)\Big).
\end{equation}
When $A_j(k, \theta)$ is continuously defined at $\kp=0$, then the above integral over circles with vanishingly small radii converges to $0$. However, since $\psi_j(k, \theta)$ is an eigenvector of $H(k, \theta)$ we observe that the above formula remains valid if $\psi_j(k, \theta)$ is (globally) gauge-transformed to, e.g., $e^{im\theta}\psi_j(k, \theta)$. This flexibility proves convenient in practice.

We can further simplify \eqref{eq:Stokes} by taking advantage of the isotropy of $\hat H(k, \theta)$. We know that
\[
\hat H(k, \theta) = U(\theta) \hat H(k, 0) U^*(\theta)
\]
for $U(\theta)$ a family of $\Cm^n-$ unitary transformations. This implies that the branches of spectrum $\lambda_j(k, \theta)=\lambda_j(k)$ are independent of $\theta$ and we may choose the eigenvectors as $\psi_j(k, \theta)=U(\theta)\psi_j(k, 0)$. Thus we obtain:
\pagebreak

\begin{align*}
A_j(k, \theta) = \Big(U(\theta)\psi_j(k, 0), d\big(U(\theta)\psi_j(k, 0)\big)\Big)\qquad\qquad\qquad \qquad\qquad\qquad\qquad \nonumber \\
\qquad= \Big(\psi_j(k, 0), U^*(\theta)dU(\theta)\psi_j(k, 0)\Big) + \big(\psi_j(k, 0), \frac d{dk}\psi_j(k, 0)\big)dk.\nonumber
\end{align*}

Isotropy, or invariance by rotation, implies that $U^*(\theta)dU(\theta)= U'(0)d\theta$ is independent of $\theta\in[0,2\pi)$ (which can also be verified directly from \eqref{eq:U}). In this setting, we thus obtain, since the integrals in \eqref{eq:Stokes} are taken on constant $k$ curves, that
\begin{equation}\label{eq:simpChern}
    \mC[\Pi_j] = i\big[ \lim_{k\to\infty} (\psi_j(k, 0),U'(0)\psi_j(k, 0)) - \lim_{k\to0} (\psi_j(k, 0),U'(0)\psi_j(k, 0)) \big].
\end{equation}
In other words, all we need to compute is $\psi_j(k, 0)$ with $k\to0$ and $k\to\infty$, which may be obtained analytically, to obtain ${\rm BDI}(u^S, u^N, \gamma) = \fC[\Pi^S, \Pi^N] = \mC[\Pi^S]-\mC[\Pi^N] = \sum_{j\le m} \left(\mC[\Pi_j^S]-\mC[\Pi_j^N]\right)$ by the (non-trivial) additivity of Chern numbers. Explicit expressions for the projectors as $k\to\infty$ are also necessary in order to verify gluing conditions recalled below in \eqref{eq:gluing}.

\paragraph{Eigenvectors of $\hat H_B^{N/S}$ as $k \to \infty$.}
Suppose that $\psi(k, 0) = (a_1, b_1, a_2, b_2, ..., a_m, b_m)^t$ is an eigenvector of $\hat H$ with eigenvalue $\lambda$; $\hat H(k,0) \psi = \lambda \psi(k, 0)$. First consider the $2j$ and $2j-1$ lines of the eigenvalue equation:
\begin{equation}\label{eq:eig_eq_j}
\begin{cases}
a_ju_j + b_j k + \gamma b_{j-1} = \lambda a_j\\
b_j u_j + a_jk + \gamma a_{j+1} = \lambda b_j.
\end{cases}
\end{equation}
By eliminating $\lambda -u_j$ we then obtain:
\[
(b_j^2-a_j^2)k = \gamma(a_{j+1}a_j - b_jb_{j-1}).
\]
Clearly the right hand side is finite so as $k \to \infty$ we must have that $b_j^2-a_j^2 = O(1/k)$ so that $\lim_{k\to \infty} b_j = \pm\lim_{k\to\infty} a_j$. WLOG we assume $\lambda > 0$ and $\lim_{k\to \infty}a_j = \lim_{k\to \infty}b_j = c_j$ (if $b_j = -a_j$ then $\lambda <0$ as $k \to \infty$). Now the first line of \eqref{eq:eig_eq_j} can be rearranged to get:
\begin{equation}\label{eq:cj}
c_{j-1} = c_j\lim_{k\to \infty}\frac {\lambda-(u_j+k)}\gamma.
\end{equation}
Assuming that $\lambda = k  + O(1)$ so that the right hand side can be finite, we have a system of equations that can be solved for $c_j$'s and $\lambda$. However, an explicit expression for these eigenvectors (particularly for arbitrary $m$) requires considerable additional algebra which we would like to avoid and is in fact unnecessary. We have now shown that eigenvectors associated to positive eigenvalues have the form:
\[
\psi_j^+ = \sum_{j = 1}^m c_j \phi_j, \qquad
\phi_j = \frac 1{\sqrt{2}}(\hat e_{2j-1} + \hat e_{2j}).
\]
From Theorem 2.2 we know that there must be $m$ positive eigenvalues for $k\ne 0$, and furthermore since $\hat H$ is Hermitian $\{\psi_j^+\}_{j = 1}^m$ is an orthonormal set. Now clearly from the above expression for $\psi_j^+$, span$\{\phi_j\}_{j = 1}^m$ = span$\{\psi_j^+\}_{j = 1}^m$. Therefore setting:
\[
\Pi_j^h = |\psi_j^h\rangle\langle\psi_j^{h}|,\qquad \tilde \Pi_j^h = |\phi_j^h\rangle\langle \phi_j^h|
\]
we have:
\[
I-\Pi^h = \sum_{j\ge m+1}\Pi_j^h = \sum_{j\ge m+1}\tilde \Pi_j^h.
\]
Since $\fC[\Pi^S, \Pi^N] =  \mC[\Pi^S] -\mC[\Pi^N]$, the BDI is independent of the choice of basis in which we choose to express $\Pi^h$. In other words, since we have found that $\text{Ran}(\Pi^h)$ has an orthonormal basis $\{\phi_j\}_{j = 1}^m$ we are free to use $\phi_j$ in place of $\psi_j^+$ in \eqref{eq:simpChern} provided we are summing over $j = \{1, ..., m\}$ to obtain $\mC[\Pi^S], \mC[\Pi^N]$.

If we are interested in explicitly verifying gluing conditions \eqref{eq:gluing} we may consider an alternate approach. Instead consider replacing $k\sigma_1$ in the $j$th 2x2 diagonal block with $\alpha_jk\sigma_1$, where $\alpha_j$'s are all distinct. Equations \eqref{eq:eig_eq_j}\eqref{eq:cj} are correspondingly modified by $k\to \alpha_jk$. Then if $\lim_{k\to\infty} \lambda -(u_j +\alpha_jk)$ is finite for some $j$, it is necessarily infinite for all $i\ne j$ since $\lim_{k\to \infty} (\alpha_j-\alpha_i)k = \pm \infty$. Therefore from \eqref{eq:cj} the only possible eigenvalues and eigenvectors are asymptotically as $k\to \infty$:
\[
\lambda_j^\pm \approx u_j\pm\alpha_jk, \qquad \psi_j^\pm = \frac 1{\sqrt{2}}(\hat e_{2j-1} \pm \hat e_{2j}).
\]

 It remains to justify the replacement of each $k\sigma_1$ by $[0,1]\ni t\mapsto((1-t)+t\alpha_n)k\sigma_1$, which is continuous in $t$. For each value of $t$, the corresponding interface Hamiltonian satisfies [H1] in \cite{bal2022topological,quinn2024approximations} so that the BDI is given by the Fedosov-H\"ormander formula \cite{bal2022topological} (Theorem 4.13), an integer-valued topological winding number that is clearly independent of $t\in[0,1]$\cite{bal2023topological}. These expressions allow us verify gluing conditions \eqref{eq:gluing} without finding eigenvectors which in fact rely non-trivially on $(u, \gamma)$, although solving \eqref{eq:cj} explicitly is also possible to verify \eqref{eq:gluing}, as done in \cite{bal2023mathematical} for the 2-layer case.

Regardless of the approach taken, we have shown that the subspace spanned by the positive (or negative) eigenvectors as $k\to \infty$ does not depend on the parameters $(\gamma, u)$, and therefore the term $\lim_{k\to \infty} (\psi(k, 0), U'(0)\psi(k, 0))$ in \eqref{eq:simpChern} after summing over $\sum_{j\ge m+1}\mC[\Pi_j^h]$ or $\sum_{j\le m}\mC[\Pi_j^h]$ is identical for any phase (for the latter simply repeat the preceding arguments with $\lambda < 0$). In particular denoting $\sum_{j\le m}\mC[\Pi_j^h] = c_\infty$ we may add a global gauge transformation $e^{-ic_\infty \theta}$ to the eigenvectors $\psi_j$ so that the contribution of the $k\to \infty$ terms to $\mC[\Pi^h]$ is in fact 0 for all phases.
\paragraph{Eigenvectors of $\hat H_B^{N/S}$ as $k\to 0$.} For $k = 0$ we note the following eigenvectors:
\[
E_0 = u_1 = -\frac u2, 
\;\;
\psi_0(0, 0) = (1, 0, 0, ..., 0)^t
\]
\[
E_m = u_m = \frac u2, 
\;\;
\psi_m(0, 0) = (0, 0, 0, ..., 0, 1)^t
\]
\[
E_n^\pm = \frac{u}{m-1}\Big[\big(n-\frac m2\big)\pm \delta(\gamma, u)\Big], \;\;\;
\psi_n^+(0, 0) = c\hat e_{2n} + s\hat e_{2n+1},\;\;\;\psi_n^-(0, 0) = c\hat e_{2n+1} -s\hat e_{2n},
\]
for $ 1\le n \le m-1$ and $c^2 + s^2 =1$. We obtain from \eqref{eq:U}:
\begin{equation}\label{eq:R}
    U'(0) =  i \;{\rm Diag}(0,1,1,2,2,\ldots,m-1,m-1,m).
\end{equation} 
Note that for $0 \le n\le m$ we have that:
\[
-i\lim_{k\to 0}\big(\psi_n^\pm(k, 0), U'(0)\psi_n^\pm(k, 0 )\big) = n.
\]

\paragraph{Computation of BDI:}Recall from \eqref{eq:phase_def} that $\hat H_B(\bk, \gamma, u)$ is in phase $j$ when $\delta_{j-1} < \delta(\gamma, u) < \delta_j$. Denote $\Pi^{(j)}$ by \eqref{eq:proj} with $\hat H_B(\bk)$ in phase $j$. First, we compute $\mC[\Pi^{(\pm1)}]$ for $m$ even, then extend to arbitrary phase $j$ and odd $m$. For $m$ even $\delta_1 = 1$ so that in phase 1 $E_n^\pm < 0$ for $0 \le n \le \frac m2 -1$ and $E_{m/2}^- < 0$ since $0< \delta(\gamma, u) < 1$. Therefore we obtain:
\[
 \mC[\Pi^{(1)}] =-i\lim_{k\to 0}\sum_{E_n^\pm < 0}\big(\psi_n^\pm(k, 0), U'(0)\psi_n^\pm(k, 0 )\big) = 2\sum_{n = 1}^{\frac m2-1} n + \frac m2 = 
\frac m2\big(\frac m2-1\big) + \frac m2 = \frac {m^2}4.
\]
For phase -1 the results are identical except that $u<0$ so that $E_n^\pm < 0$ for $\frac m2+1 \le n \le m$ and $E_{\frac m2}^+ < 0$ which gives:
\[
 \mC[\Pi^{(-1)}]=-i\lim_{k\to 0}\sum_{E_n^\pm < 0}\big(\psi_n^\pm(k, 0), U'(0)\psi_n^\pm(k, 0)\big) = \frac m2 + 2\Big(\sum_{n = \frac m2+1}^{m-1} n\Big) + m = \frac{3m^2}{4}. 
\]
In order to enforce the symmetry $\mC[\Pi^{(j)}] = -\mC[\Pi^{(-j)}]$ (which we may do WLOG since our final result involves differences of $\mC[\Pi^{(j)}]$'s) we apply another global gauge transformation of $\psi_j(k, \theta) \to e^{-i\frac {m^2}4\theta}\psi_j(k, \theta)$ so that:
\[
\mC[\Pi^{(\pm1)}] = \mp \frac {m^2}{4}.
\]
Now notice that if $\delta_{j-1} < \delta(\gamma, u) < \delta_{j}$ then for $u > 0$, $E_{\frac m2 + l}^- <0$ and $E_{\frac m2 - l}^+ > 0$ for $1\le l \le j-1$. Therefore for $j>1$ we can obtain $\mC[\Pi^{(j)}]$ by subtracting the $E^+_{\frac m2-l}$ and adding the $E_{\frac m2 + l}^-$ curvature terms from $\mC[\Pi^{(1)}]$ for $1\le l\le j-1$:
\[
\mC[\Pi^{(j)}] = \mC[\Pi^{(1)}] +\sum_{l = 1}^{j-1}\left( \left(\frac m2 +l\right)-\left(\frac m2 -l\right)\right) = -\frac {m^2}4 + 2\sum_{l = 1}^{j-1}l = -\frac {m^2}4 + \delta_j(\delta_j-1).
\]
Similarly for $u < 0$ we get $E_{\frac m2 + l}^- >0$ and $E_{\frac m2 - l}^+ < 0$ for $1\le l \le j-1$ so that for $j<1$:
\[
\mC[\Pi^{(j)}] = \mC[\Pi^{(-1)}] -\sum_{l = 1}^{j-1}\left( \left(\frac m2 +l\right)-\left(\frac m2 -l\right)\right) = \frac {m^2}4 - \delta_j(\delta_j-1) = -\mC[\Pi^{(-j)}].
\]
Performing similar calculations with $m$ odd we find, with a global gauge transformation of $e^{-i(\frac {m^2}2 -m)\theta}$, that:
\[
\mC[\Pi^{(1)}] = -\mC[\Pi^{(-1)}] = -\frac {m^2}4 + \frac 34 = -\frac {m^2}4 + \delta_1(\delta_1-1)
\]
and subsequently for $j >1$:
\begin{align*}
    \mC[\Pi^{(j)}] = -\mC[\Pi^{(-j)}] = -\frac {m^2}4 + \frac 34 + \sum_{l = 1}^{j-1}\left(\left((\lceil \frac m2\rceil +l\right)-\left(\lfloor \frac m2\rfloor -l\right)\right) \qquad\qquad \nonumber\\\qquad
    = -\frac {m^2}4 + (j^2-\frac 14) = -\frac {m^2}4 + \delta_j(\delta_j-1). \nonumber 
\end{align*}
Thus for all $m$ and all phases $j \in \{\pm1, \pm2, ..., \pm \lfloor \frac m2\rfloor\}$ we get:
\begin{equation}\label{eq:bulk_inv}
\mC[\Pi^{(j)}] = -\text{sgn}(j)\left(\frac {m^2}{4}-\delta_j(\delta_j-1)\right) 
\end{equation}
Therefore noting that BDI$_{kj} = \fC[\Pi^{(k)}, \Pi^{(j)}] = \mC[\Pi^{(k)}]-\mC[\Pi^{(j)}]$ we obtain \eqref{eq:BDIj} by substituting \eqref{eq:bulk_inv} into \eqref{eq:ChernBDI}.

\section{Numerical simulations}
\label{sec:num}

We illustrate the theoretical results presented in Theorems \ref{thm:gap} and \ref{thm:class} by a number of numerical simulations. 

\paragraph{Bulk spectrum simulations.}
We start with cross sections of the rotationally invariant bulk spectrum of the 5-layer system \eqref{eq:HB} (involving diagonalization of $10\times10$ matrices performed in MATLAB) in Figure \ref{fig:5-layer_bulk} for various values of $u$ assuming a fixed coupling constant normalized to $\gamma = 1$. For the low-potential (relative to $\gamma$) case we recover the band structure noted in \cite{zhang2011spontaneous} for tri-layers which exhibits two minima (maxima) in the bands closest to $E = 0$. In the high-potential case, we recover the pattern discussed in \cite{bal2022multiscale} where for $m$ layers, $m$ minima (maxima) are observed in the bands closest to $E = 0$ and the distance from $E = 0$ scales as $(\gamma/u)^{2l}$ for $l = \{1, ..., \lceil m/2\rceil\}$. For intermediate values of $u$, we observe an intermediate number of minima close to $E = 0$ which varies from 2 to $m$ as $u$ increases. Importantly, the size of the global band gap scales  proportional to $(\gamma/u)^{2m}$ for $u$ sufficiently large as described in \cite{bal2022multiscale} so that the global gap can be arbitrarily small as $u, m \to \infty$, even if the gap is guaranteed to exist by Theorem \ref{thm:gap}. This is illustrated in Figure \ref{fig:5-layer_bulk}(d) where the gap is barely discernible for the highest $|\bk|$ at which a minimum occurs.

The theory in \cite{bal2022multiscale} ensures that the smallest spectral gap appears at the large values of $|\bk|$, which is somewhat reassuring given that high-wavenumber oscillations may be able to be disregarded in real systems. However, for some parameter regimes close to critical values, we find that arbitrarily small spectral gaps can also be found for a range of intermediate $|\bk|$ which are not too large nor exceedingly close to $|\bk| = 0$, where we know that band crossings occur according to \eqref{eq:ucrit}. In Figure \ref{fig:9-layer} we demonstrate this for $m = 9$, where a spectral gap of $\approx10^{-3}$ is seen at $|\bk| \approx \pm 0.32$, creating an obstacle to observing edge states experimentally for higher values of $m$. This finding may also call into question whether the topological properties of the system are robust to perturbations due to the fact that our analysis relies heavily on a global gap which could be closed by exceedingly small perturbative effects at these parameter values. 
\begin{figure}[b!]
    \centering
    \includegraphics[width = .8\textwidth]{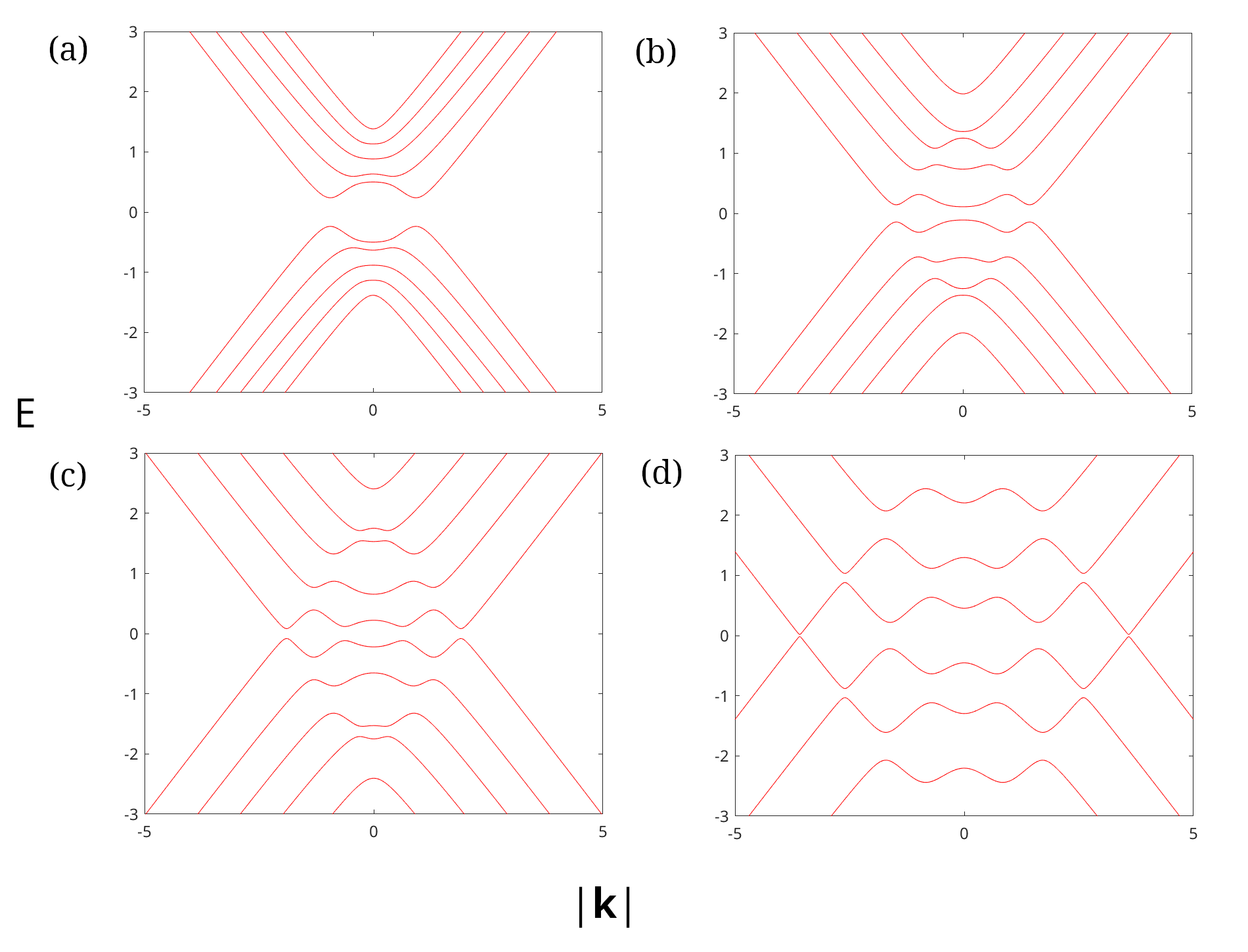}
    \caption{Bulk spectrum for 5-layer model with $\gamma = 1$ and $u = (1, 2.5, 3.5, 7)$ for (a), (b), (c), and (d) respectively. For the low-potential case we recover the ``Mexican hat" shape \cite{min2008electronic, zhang2010band} for the bands nearest $E = 0$ and for the high-potential case we note 5 minima (maxima) in the positive (negative) band closest to $E = 0$\cite{bal2022multiscale}. Intermediate values of $u$ produce an intermediate number of minima.}
    \label{fig:5-layer_bulk}
\end{figure}
\begin{figure}[t!]
    \centering
    \includegraphics[width = \textwidth]{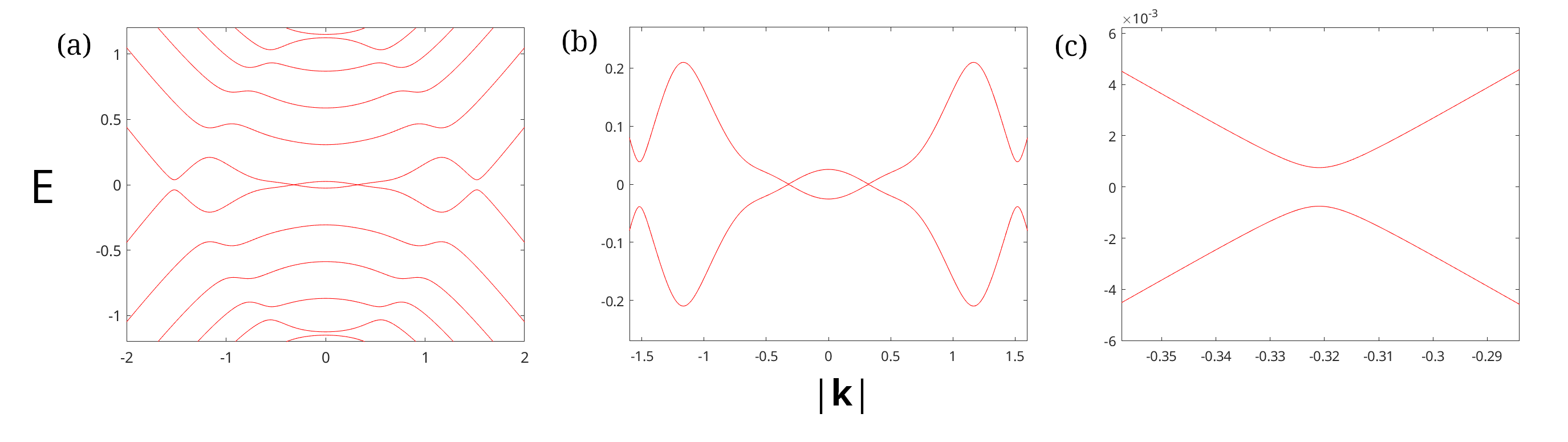}
    \caption{(a) Bulk spectrum for $m = 9$ layers. As number of layers increases the global band gap near phase transitions becomes increasingly small even for intermediate values of $|\bk|$. (b) shows the branch closest to $E = 0$ and (c) shows detail of the minimum in said branch closest to $k_x = 0$. Spectrum shown is for $(\gamma, u) = (1, 2.25)$.}
    \label{fig:9-layer}
\end{figure}

\paragraph{Edge spectrum simulations.}
In order to calculate the spectrum of interface Hamiltonians numerically, we assume that $H$ is invariant in the x-direction and $u(y)$ varies between $u^S$ and $u^N$ in the region $-R\le y\le R$ as described above. Taking the Fourier transform in $(x, t)\to (k_x, E)$ under these assumptions gives an interface Hamiltonian:
\[
 H_I(k_x) = \begin{pmatrix} u_1(y)+ k_x\sigma_1 -i\sigma_2\partial_y& B & 0 & \ldots & 0 \\
   B^* & u_2(y) + k_x\sigma_1-i\sigma_2\partial_y & B & \ddots &  \vdots \\
   0 & B^* & \ddots  & \ddots & 0 \\ 
   \vdots  & \ddots & \ddots & \ddots & B\\
   0 & \ldots & 0 & B^* & u_m(y)+ k_x\sigma_1-i\sigma_2\partial_y
   \end{pmatrix}
\]
and the eigenvalue problem in one dimension
\[
H_I(k_x)\psi(y) = E\psi(y).
\]
Now denote $S_2 = (I_m\otimes\sigma_2)$ and $C(k_x) = H_I+iS_2\partial_y$. Then the eigenvalue problem above is equivalent to:
\begin{equation}\label{eq:ode}
-iS_2\big(E-C(k_x)\big)\psi(y) =: A(k_x, E)\psi(y)= \partial_y \psi(y).
\end{equation}
Assuming that for large enough $|y|$, the coefficients of $A(k_x, E)$ are constant and given by the constant-coefficient operator $A^{N/S}(k_x, E)$ for $y\geq R$ and $y \leq -R$ respectively, then for any fixed $(k_x, E)$ we can determine the bulk modes by simply diagonalizing $A^{N}, A^S$. In general, however, an explicit solution for $(\psi(y),E)$ is not feasible analytically for the ODE \eqref{eq:ode}. A numerical ODE solver is used to solve \eqref{eq:ode} for $\psi(y=0)$ by using the eigenvectors of $A^N$ and $A^S$ as initial conditions at $y = \pm R$ respectively and numerically solving for $\psi(y)$ for $-R< y < 0$ and $0 < y < R$ respectively. Exponentially increasing eigenvectors from the bulk spectrum of $A^N$ and $A^S$ are eliminated as non-physical. If the subspaces of valid (non-exponentially increasing) eigenvectors of $A^N, A^S$ intersect non-trivially once evaluated at $y = \pm 0$ then there exists a valid solution of \eqref{eq:ode}. More precisely, orthogonal projectors $P^N, P^S$ onto the subspaces of valid eigenvectors of $A^N, A^S$ evaluated at $y = 0$ are formed, and a $(k_x, E)$ pair is accepted if the largest eigenvalue of $P^NP^S$ is within a small tolerance of 1 ($<10^{-4}$ for Figures 3-6). When compared with widely-used finite difference methods \cite{fu2021topological, bal2024topological, bal2022multiscale} this method has no need for periodization of the domain and therefore does not require the heuristic elimination of modes \cite{bal2024topological, bal2022multiscale} or alternatively the necessity of modeling two equal but opposite transitions \cite{fu2021topological, zhang2011spontaneous}. This method shares these advantages with two previously-proposed approaches\cite{colbrook2019compute, colbrook2023computing} but may also be applied to continuum models as well as periodic lattice models. Comparison with finite-difference methods (not displayed here) does however show that the two methods agree using a fine enough mesh with finite differences where spurious modes are eliminated.

Figure \ref{fig:5-layer_phase} validates the values of the BDI given in Theorem \ref{thm:class} for a symmetric transition $u^N = -u^S > 0$ when $3\leq m\leq 6$ layers and for all possible values of $\delta_j$. In this case BDI$_{-j,j} = \frac {m^2}2-2\delta_j(\delta_j-1)$. As shown in Theorem \ref{thm:class} and Fig. \ref{fig:predict}, $m =3$ demonstrates only one distinct topological phase transition. For $m = 4, 5, 6$ we derive from \eqref{eq:ucrit} the critical values $u_c = 2\sqrt{2} \gamma$ ($m = 5$), $u_c = 2\sqrt{3} \gamma$ ($m = 4$), and $u_c = (10\gamma/\sqrt{3}, 10\gamma/\sqrt{15})$ ($m = 6$). Figure \ref{fig:5-layer_phase} contains numerically calculated spectra for all symmetric transitions between phases divided by these critical values assuming $\gamma = 1$. For the asymmetric case the number of possible transitions increases proportional to $m^2$, however in Figure \ref{fig:5-layer_asym} we show the remaining possible transitions for $m = 5$ between phases +1 and +2 and -1 and +2 respectively, which produce BDI's of 3 and 8. We also validate in Figure \ref{fig:5-layer_trans} using 5 layers that band crossings only occur at $k = 0$ and show the transition from 5 to 11 edge modes at the critical value $u = 2\sqrt{2}\gamma$. All results are in perfect agreement with the results of Theorems \ref{thm:gap} and \ref{thm:class}.

\begin{figure}[t!]
    \centering
    \includegraphics[width = \textwidth]{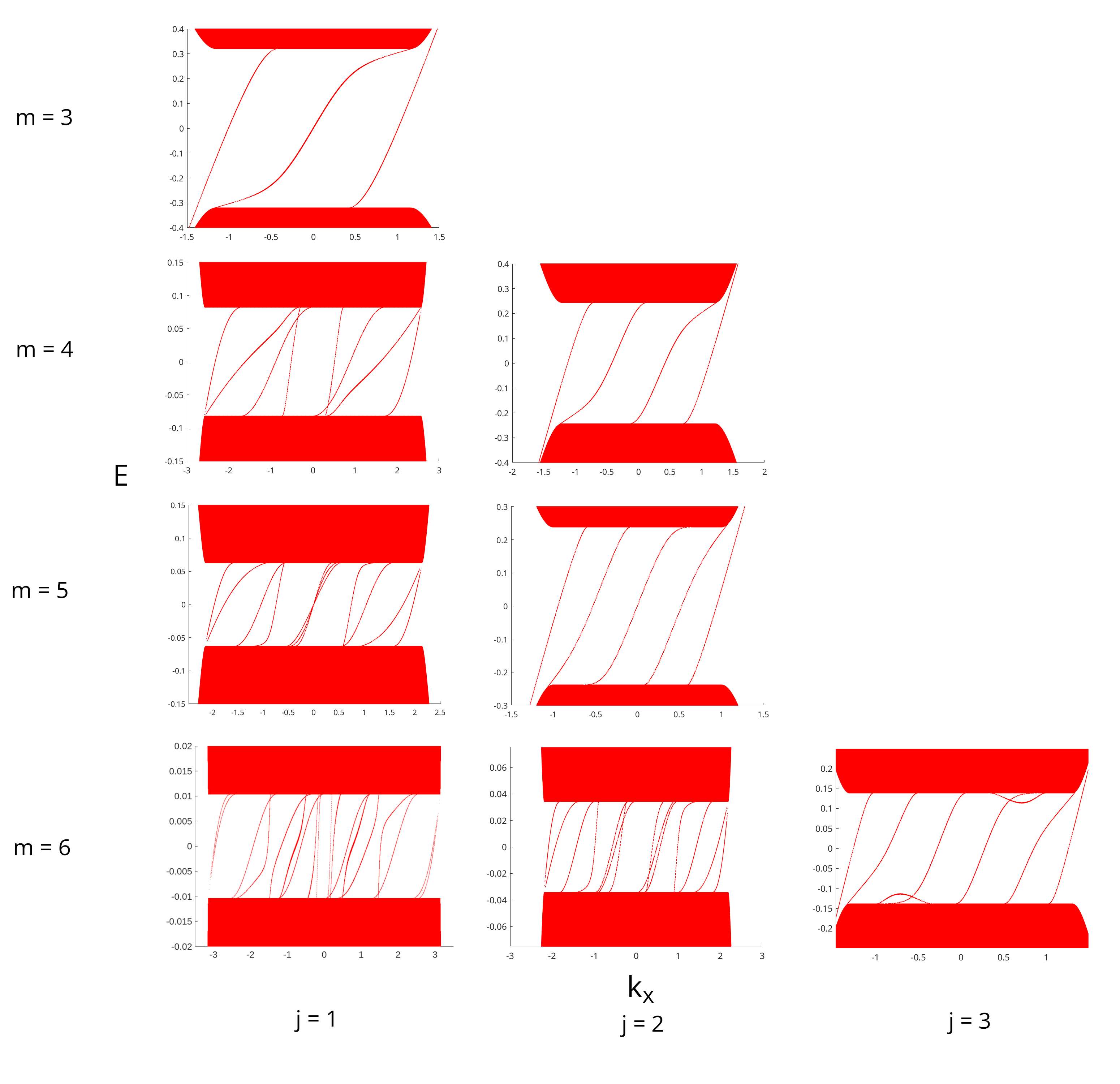}
    \caption{Numerically calculated spectra for $m = \{3, 4, 5, 6\}$ for symmetric transitions between the two phases $-j$ and $j$ for $j\in \{1, 2, 3\}$. We assume $\gamma = 1$ for simplicity and calculate the spectrum of $H_I$ for $u = 2$ $(m = 3)$, u = (5, 2) ($m = 4$), $u = (4, 1.2)$ ($m = 5$), and $u = (6, 4, 2)$ ($m = 6$). The approximate critical values from \eqref{eq:ucrit} are $u_c \approx 3.46$ (m = 4), $u_c \approx 2.83$ (m = 5), and $u_c \approx (2.58, 5.77)$ ($m = 6$). Note that BDI's for symmetric transitions are denoted by $^\dag$ in Figure \ref{fig:predict}. The number of edge states agrees with Figure \ref{fig:predict} for $m = 3, 4, 5, 6$ and all respective phases.}
    \label{fig:5-layer_phase}
\end{figure}

\begin{figure}[t!]
    \centering
    \includegraphics[width = \textwidth]{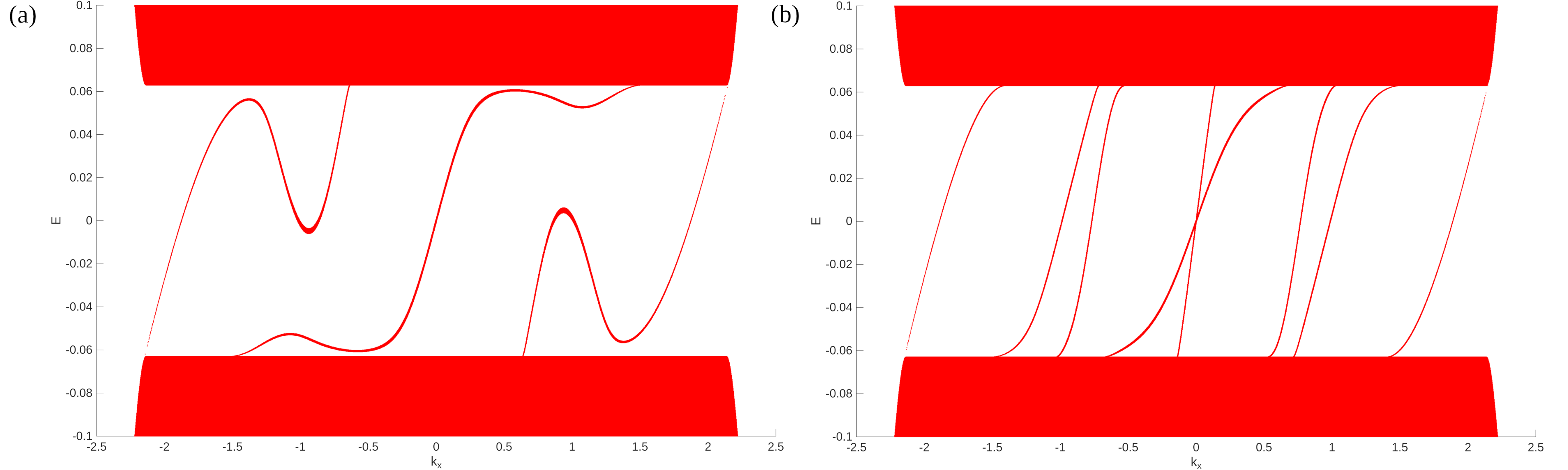}
    \caption{Illustration of the asymmetric transitions in $u$ for $m = 5$ layers. (a) shows a transition from phase +1 to +2 where $u$ varies from $u^S = 2$ to $u^N = 4$. (b) shows a transition from phase -1 to phase 2 where $u$ varies from $u^S = -2$ to $u^N = 4$. Both transitions are for $\gamma = 1$, giving a critical value of $u_c = 2.83$ as in the previous figure.}
    \label{fig:5-layer_asym}
\end{figure}

\begin{figure}[t!]
    \centering
    \includegraphics[width = \textwidth]{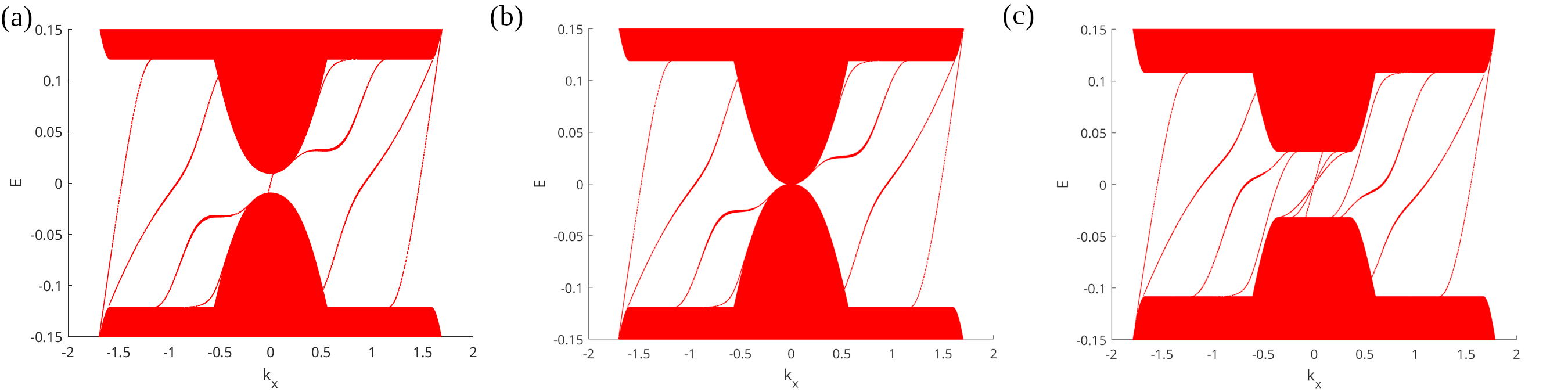}
    \caption{Illustration of the transition through the critical value $u =2\sqrt{2}\gamma$ for five-layer ABC-stacked graphene (m = 5). Potential differences of $u = (2.8, 2 \sqrt{2} \approx 2.83, 3)$ are shown respectively for (a), (b), (c), with $\gamma = 1$.}
    \label{fig:5-layer_trans}
\end{figure}

\begin{figure}[t!]
    \centering
    \includegraphics[width = .75\textwidth]{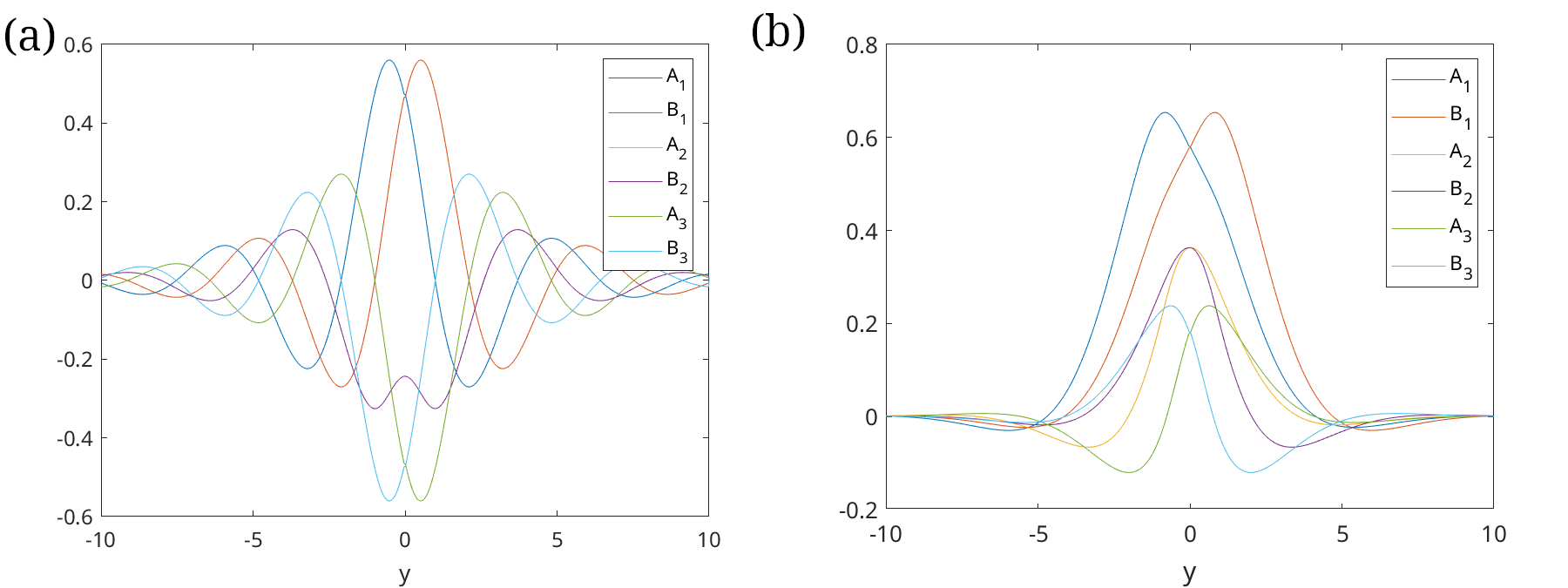}
    \caption{Eigenvectors of edge states in the 3-layer case ($(u, \gamma) = (2, 1)$) at $E = 0$. (a) is the eigenvector as a function of $y$ at $(k_x, E) = (0, 0)$ and (b) at $(k_x, E) = (-1, 0)$. The edge is located within $-1 < y < 1$. $A_n$ and $B_n$ labels correspond to the $A$ and $B$ sub-lattice components of layer $n$ following the notation of \cite{zhang2010band}. Polarization is seen biased toward the outer layers.}
    \label{fig:eigs}
\end{figure}

One particularly interesting feature shown in Figure \ref{fig:5-layer_phase} and \ref{fig:5-layer_asym} is the flattening of the bulk spectrum above and below the band gap. This feature can in fact be straightforwardly deduced from the rotational invariance of the bulk Hamiltonian. Note that Figure \ref{fig:5-layer_bulk} shows just a cross-section in $|\bk|$ of the bulk spectrum. First, fix any $E$ which is not in the bulk band gap, and denote $k_{max} = \arg\max_{|\bk|} \{E_j(\bk) = E, j\in \{1, ..., 2m\}\}$. Then clearly for any $k_x$ with $|k_x|\le k_{max}$ we can find a $\bk$ with $|\bk| = \sqrt{k_x^2+k_y^2} = k_{max}$ for some $k_y \in [0,k_{max}]$ so, due to the rotation invariance of the spectrum, there is in fact a bulk mode corresponding to any $k_x \le k_{max}$ for each $E$ not in the bulk band gap. Therefore the bulk spectrum of $H_I(k_x)$ must be flat in between any two peaks in the bulk spectrum, as we observe in Figures \ref{fig:5-layer_phase}, \ref{fig:5-layer_asym}, \ref{fig:5-layer_trans}, and \ref{fig:5-layer_eps}.

Finally, the eigenvectors for two edge modes are illustrated in Figure \ref{fig:eigs} for the 3-layer case. These eigenvectors confirm that the edge modes are concentrated around $y = 0$ and show that the polarization of these modes favors the outer layers. 
\section{Discussion}
\label{sec:conclu}
The result of Theorem \ref{thm:class} gives all possible topological values of the quantum anomalous Hall effect in an idealized macroscopic (effective) model of multi-layer rhombohedral graphene and Floquet topological insulators in a setting where the gate potential varies between two bulk regions. In the latter application, we retrieve the topological invariant obtained for large values of the driving laser frequency (corresponding to a small coupling constant $\gamma$) \cite{bal2022multiscale}. In the former application, we retrieve the topological invariant obtained for small displacement field \cite{han2024large,zhang2011spontaneous} compared to the coupling constant $\gamma$. 

We also note that the spatially varying gate potential is just one possible way to introduce non-trivial topological phases. One other possibility is the case that across a large enough length scale the energetically preferable stacking order changes continuously from ABC to CBA, as considered in \cite{bal2023mathematical}. This transition is modeled by a continuous transition of $B \to B^*$. Substituting the eigenvectors of $\hat H(u, \gamma A^*)$ for $\hat H(-u, \gamma A)$ in Section \ref{sec:proof} using the unitary relations derived in \eqref{eq:symtauA} it is clear that the BDI's for this transition are identical a symmetric transition from phase $-j$ to $j$ as derived in Section \ref{sec:proof}. Figure \ref{fig:5-layer_eps} shows for the 5-layer case that indeed these two transitions are topologically equivalent. Indeed considering that $B$ and $B^*$ are distinct topological phases for each value of $u$ the number of topological phases allowing both the transition matrix $B$ and gate potential $u$ to vary spatially is $4\lfloor\frac m2\rfloor$. Bulk invariants for all these phases may be calculated similarly to in section \ref{sec:proof}.d and subsequently subtracting any two of these will give a BDI for a transition between any two of the $4\lfloor \frac m2\rfloor$ phases. One could also allow the possibility that a non-uniform electric field produces potential differences that are not constant between any two adjacent layers or adding insulating layers between graphene layers to introduce non-uniform coupling constants $\gamma$, which is the subject of ongoing work.

\begin{figure}[t!]
    \centering
    \includegraphics[width = .8\textwidth]{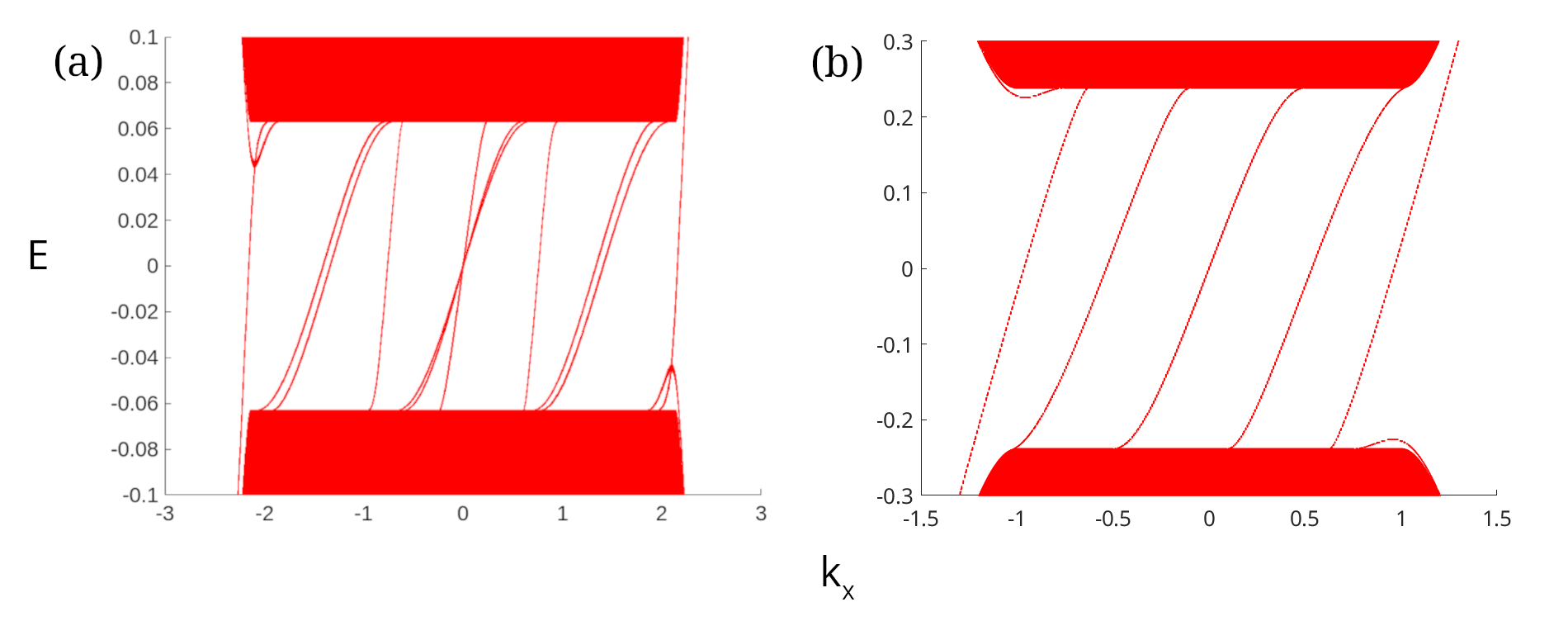}
    \caption{Edge spectrum for 5-layer case with $A\to A^*$ transition. Parameter values used are $u = (4, 1.2)$ for (a), and (b) respectively and $\gamma = 1$. Comparison with Figure \ref{fig:5-layer_phase} $(m = 5)$ shows that the $A\to A^*$ transition is topologically equivalent to the $u\to -u$ transition.}
    \label{fig:5-layer_eps}
\end{figure}

In the RHG application, we assume here a model \eqref{eq:H} that is both spin- and valley-polarized. While the obtained numbers do not depend on spin, they depend on the valley index $\tau=\pm1$. The total QAH therefore vanishes unless valley polarization may be obtained experimentally, which is achieved by means of an appropriate spin orbit coupling\cite{han2024large}. It is unclear whether such polarization may still be achieved for the large displacement fields (large values of $u$) necessary to observe a phase transition (i.e., $\delta<\delta_{\lfloor\frac m2 -1\rfloor}$).

The theoretical results, and in particular the presence of a gap at $E=0$, depend on the specific structures of $H$ in \eqref{eq:H}. More general interlayer couplings such as those proposed in \cite{han2024large,zhang2011spontaneous} do not modify the elliptic nature of the operators and the bulk-edge correspondence still applies\cite{bal2023topological,quinn2024approximations}. Therefore, any continuous deformation of the model in \eqref{eq:H} with bulk phases that does not close the gap at $E=0$ will generate an edge asymmetry for which the results and explicit indices of Theorem \ref{thm:class} apply. While the results of Theorem \ref{thm:gap} guarantee that gap crossings may only occur at $E=0$, we found examples of values of $(\gamma,u)$ such that the spectral gap at values $k\not=0$ may be significantly smaller than the gap at $k=0$; see for instance Fig. \ref{fig:9-layer}. When considering interactions between more distant layers, as in \cite{zhang2010band}, rotational symmetry which has been essential to our analysis here breaks down and the perturbative effects of further off-diagonal elements have yet to be explored. It is therefore unclear whether such gaps persist except for very small perturbations of the model \eqref{eq:H}. 

In summary, we have found all the topological phases and bulk difference invariants of an effective model of coupled 2-d Dirac systems which has applications in gated ABC-stacked rhombohedral graphene with an arbitrary number of layers \cite{zhang2010band, min2008electronic, zhang2011spontaneous} and Floquet topological insulators \cite{bal2022multiscale}. Due to the ellipticity of the problem the bulk-edge correspondence holds \cite{bal2023topological, bal2022topological} and we verify by numerical spectral calculations the number of predicted edge modes for all topological phases in the $m = \{3, 4, 5, 6\}$ cases. For $u < \gamma$ we recover the results of \cite{zhang2011spontaneous} for the QAH effect in RHG and for $u \gg \gamma$ the results of \cite{bal2022multiscale} for Floquet topological insulators. Our results show that for an $m$-layer model with a spatially varying gate potential there are in fact $2\lfloor\frac m2\rfloor$ topological phases and the asymmetric charge transport at an interface between any two of these phases is quantized by the BDI in \eqref{eq:BDIj}. Although it remains to be seen if these phases are robust to perturbations in interactions beyond the nearest-layer interactions we consider here, these results show the exciting possibility of tunable QAH states in multi-layer ABC stacked graphene should experimental conditions allow high enough gate potentials to be achieved.

\section*{Acknowledgments} 
The authors acknowledge stimulating discussions with Allan MacDonald on graphene systems and Jeremy Hoskins for invaluable advice and input regarding numerical techniques. This work was funded in part by NSF grant DMS-230641 and ONR grant N00014-26-1-2017.

\section*{Author Declarations}
The authors have no conflicts to disclose.

\section*{Data Availability}
The code needed to reproduce the data in section \ref{sec:num} is available at:\\ https://doi.org/10.5281/zenodo.17156250.
\appendix
\section{Classification of elliptic Hamiltonians}\label{sec:app}
\medskip
\paragraph{Classification of Elliptic Operators.}

Consider an interface Hamiltonian $H_I$ with matrix-valued symbol $a(\bx,\kp):\Rm^{4}\to \Cm^{2m\times 2m}$ in the Weyl quantization, i.e., 
\begin{equation}\label{eq:weylquantization}
  H_I f(\bx)= (\ow a) f(\bx)  := \dint_{\Rm^{4}} \dfrac{e^{i(\bx-\by) \cdot\kp}}{(2\pi)^2}  a(\frac{\bx+\by}2,\kp) f(\by)  d\kp d\by,
\end{equation}
with $a(\bx,\kp)$ satisfying hypothesis [H1] in \cite{quinn2024approximations,bal2022topological} of order $\langle \bx\rangle$ (of order $m = 1$ in the notation of \cite{quinn2024approximations,bal2022topological}), i.e. such that the minimum singular value of $a(\bx,\kp)$, $a_{min}(\bx, \kp)$ is bounded below according to $a_{min}(\bx, \kp)\ge C|\kp|-1$ for some $C> 0$ and such that $a(\bx,\kp)=a^N(\kp)$ for $y\geq R>0$ while $a(\bx,\kp)=a^S(\kp)$ for $y\leq -R$. As such $H_I$ is an unbounded operator whose domain must be $H^1(\Rm^2; \Cm^{2m})$ or a subset thereof. Fixing $H^1(\Rm^2; \Cm^{2m})$ as the domain of $H_I$ ensures self-adjointness in the main part of the paper. We further assume that $a^{N/S}(\kp)$ are gapped for all $\kp\in\Rm^2$ in the energy interval $(-E_0,E_0)$, i.e., all of their eigenvalues lie outside of this interval. We use the notation $\bx=(x,y)$ while $\kp=(k_x,k_y)$.

Let $F=F(x) \in C^\infty(\Rm)$ be a function that depends only the spatial coordinate $x$ with  $F(x)=0$ for $x<x_0-\delta$ and $F(x)=1$ for $x>x_0+\delta$ for some $x_0\in\Rm$ and $\delta>0$. The function $F(x)$ should be interpreted as the observable quantifying the field density in the (right) half-space $x\geq x_0$. Let $0\leq\varphi\in C^\infty(\Rm)$ be a function such that $\varphi(E)=0$ for $E\leq -E_0$ and $\varphi(E)=1$ for $E\geq E_0$. Thus $\varphi'(H_I)$ defines a density of states that cannot propagate into the N and S insulating (for such energies) bulks. 

The operator $i[H_I,F]$ may be interpreted as a current operator of excitations crossing the vertical line $x=x_0$ per unit time. The expectation of this operator within the appropriate Hilbert space (in our case $H^1(\Rm^2;\Cm^{2m})$) against the density $\varphi'(H)$ is then defined as
\begin{equation}\label{eq:sigmaI}
  \sigma_I [H_I] = {\rm Tr} \ i[H_I,F] \varphi'(H_I),
\end{equation}
Tr denoting the trace in the appropriate Hilbert space. This assumes that $i[H_I,F] \varphi'(H_I)$ is a trace-class operator, which is indeed the case whenever [H1] holds \cite{bal2022topological,bal2024topological,quinn2024approximations}. Note however that if $\varphi'$ is supported on any interval containing bulk spectrum, $\sigma_I[H_I]$ is clearly not trace-class so that $\sigma_I$ necessarily characterizes a bulk band gap. This current is quantized: $2\pi\sigma_I[H_I]\in\Zm$, reflecting the topological nature of the edge states and the robustness of this current to perturbations of $a(\bx,\kp)$ satisfying condition [H1] \cite{bal2022topological,bal2024topological,quinn2024approximations}. 

The computation of such an index remains difficult in practice and typically requires a diagonalization of the operator $H_I$. An important simplification occurs when that invariant may be related to the properties of the bulk operators $H^{N/S}$ with constant coefficient symbols $a^{N/S}(\kp)$.
\medskip
\paragraph{Definition of BDI.}
Consider the  two families of self-adjoint Hamiltonians in Fourier variables $H^h(\kp)$ for $h\in\{N,S\}$ and $\kp\in\Rm^2$ with values in $\Cm^n\times\Cm^n$ and assume the following spectral decomposition
\[  H^h(\kp) = \dsum_{j=1}^n \lambda_j^h(\kp) \Pi^h_j(\kp)\]
where $\Pi^h_j(\kp)=|\psi^h_j(\kp)\rangle\langle\psi^h_j(\kp)|$ are rank-one projectors and $\lambda^h_j(\bk)$ are the corresponding eigenvalues. The spectral theorem ensures such a decomposition exists for self-adjoint operators. In the RHG application, all bands are simple and hence in fact real-analytic in $\kp$. Associated to each (arbitrary-rank) projector family $\Rm^2\ni \kp\mapsto\Pi(\kp)$ is the following integral of the associated (Berry) curvature
\begin{equation}\label{eq:curv}
    \mC[\Pi]=\frac i{2\pi} \int_{\Rm^2} {\rm tr} \Pi d\Pi \wedge d\Pi,\qquad 
d\Pi := \pdr{\Pi}{k_x}dk_x + \pdr{\Pi}{k_y}dk_y.
\end{equation}
Here ${\rm tr}$ stands for standard matrix trace. Because $\Rm^2$ is not a compact manifold, the above integral is not necessarily a (stable with respect to continuous deformations) Chern number or even necessarily integral-valued. In fact, in many applications, the above integral takes an arbitrary continuum of values (see \cite{bal2023mathematical} for a model of bilayer graphene). One reason for this fact is that the projector $\Pi$ may not have a constant value as $k\to\infty$. In such a setting, we may still be able glue two projectors, one from $H^N$ and the other one from $H^S$, by radial compactification of two Euclidean planes on the Riemann sphere \cite{bal2022topological,frazier2025topological,rossi2024topology}, and obtain a well-defined Chern number. We recall the main steps of the procedure.

Assume a spectral gap between levels $\ell$ and $\ell+1$ for both $h=N$ and $h=S$, i.e., an interval $I_\ell$ such that $\lambda^h_j(\kp)< I_\ell$ for $h\in{N,S}$ and $j\leq \ell$ while $\lambda^h_j(\kp)> I_\ell$ for $h\in{N,S}$ and $j \geq \ell+1$. In the RHG application, $\ell=m $ and $I_\ell=(-E_0,E_0)$. Associated to the spectral gap are the projectors:
\[
   P_\ell^h = \sum_{j\leq \ell} \Pi^h_j,\qquad \mW_\ell^h := \mC[P_\ell^h] = -\mC[I-P_\ell^h],\qquad I-P_\ell^h = \sum_{j\geq \ell+1} \Pi^h_j.
\]
The total curvature associated to a band may be computed as a sum either over bands below the gap or over bands above the gap. We then define the bulk-difference invariant for the gap labeled by $\ell$:
\begin{equation}\label{eq:ChernBDI2}
    \fC_\ell := \fC[P_\ell^S,P_\ell^N] := \mW_\ell^S - \mW_\ell^N  = \frac  i{2\pi} \int_{\Rm^2} {\rm tr} P_\ell^S dP_\ell^S\wedge dP_\ell^S
 - \frac  i{2\pi} \int_{\Rm^2} {\rm tr} P_\ell^N dP_\ell^N \wedge dP_\ell^N.
\end{equation}

Provided that we have the following gluing condition 
\begin{equation}\label{eq:gluing}
    \lim_{r\to\infty} P_\ell^N(r\theta) = \lim_{r\to\infty} P_\ell^S(r\theta) \quad \mbox{ for all } \theta\in \Sm^1,
\end{equation}
then $\fC_\ell$ is also a Chern number for a family of projectors defined on the sphere $\Sm^2$. Indeed, we stereographically project $P^N_\ell$ onto the upper hemisphere of $\Sm^2$ and $P^S_\ell$ onto the lower hemisphere of $\Sm^2$ while the above gluing condition ensures that the family of projectors on $\Sm^2$ is continuous across the equator. This guarantees that the {\em Bulk Difference Invariant} (BDI) $\fC_\ell\in\Zm$ \cite{bal2022topological,frazier2025topological}.

That the gluing condition \eqref{eq:gluing} holds for the RHG model \eqref{eq:weylquantization} is a consequence of ellipticity: as $|\kp|\to\infty$, the ranges of $P^N_\ell$ and $P^S_\ell$ are independent of the parameters $(\gamma,u)$ and equal. Moreover, after replacing each term $k\sigma_1$ by $\alpha_n k\sigma_1$ as was done in the proof of Theorem \ref{thm:class}, it is then straightforward to observe that the eigenvectors obtained in the limit $k\to\infty$ are independent of the displacement fields $\pm u$ and hence the gluing conditions obtained for each band separately.
\medskip
\paragraph{Bulk edge correspondence.}
We introduced two invariants in \eqref{eq:sigmaI} and \eqref{eq:ChernBDI2} above. The former is difficult to compute in general without an understanding of the spectral decomposition of $H_I$. The latter on the other hand involves a reasonably explicit integral, and as we saw in \eqref{eq:simpChern} is easily estimated when the bulk operators satisfying an invariance by rotation.

The bulk-edge correspondence is a general principle stating that the edge current asymmetry $2\pi\sigma_I$ is related to the bulk invariants by the relation
\begin{equation}\label{eq:BEC}
 2\pi\sigma_I[H_I] = {\rm BDI} = \fC[P^S_\ell,P^N_\ell]
\end{equation}
where $\ell$ is a common spectral gap of the bulk Hamiltonians $H^h$ for $h\in\{N,S\}$ and the density $\varphi'$ appearing in \eqref{eq:sigmaI} is supported in that common gap. This relation thus implies that the number of edge modes characterized by $2\pi\sigma_I$ is independent of the details of the transition between $H^S$ and $H^N$. For operators satisfying [H1], which holds for RHG, then \eqref{eq:BEC} applies \cite{bal2022topological,bal2024topological,quinn2024approximations}. As a consequence, the QAH phases described in \eqref{eq:BDIj} in Theorem \ref{thm:class} manifest themselves in the asymmetry of the edge modes observed in the numerical simulations of section \ref{sec:num}.  

\bibliography{bibRHG} 
\bibliographystyle{siam}

\end{document}